\documentclass[letterpaper, 10 pt, conference]{ieeeconf}

\makeatletter

\let\proof\@undefined
\let\endproof\@undefined
\makeatother

\usepackage{subfiles}
\usepackage{amssymb,amsmath}
\usepackage{amsthm}
\usepackage{comment}
\usepackage{graphicx}
\usepackage{pgfplots}
\pgfplotsset{compat=1.17}
\usepackage[ruled,vlined,linesnumbered]{algorithm2e}
\usepackage{float}

\newtheorem{theorem}{Theorem}[section]
\newtheorem{proposition}[theorem]{Proposition}
\newtheorem{corollary}[theorem]{Corollary}
\newtheorem{lemma}[theorem]{Lemma}

\theoremstyle{definition}
\newtheorem{definition}[theorem]{Definition}

\theoremstyle{remark}
\newtheorem{remark}[theorem]{Remark}

\theoremstyle{definition}
\newtheorem{exmp}[theorem]{Example}

\newcommand\subscr[2]{#1_{\textup{#2}}}
\newcommand\upscr[2]{#1^{\textup{#2}}}

\newcommand{\alphaopt}[1]{\upscr{\alpha_{#1}}{opt}}

\newcommand{\alphapes}[1]{\upscr{\alpha_{#1}}{pes}}
\newcommand{\alphaoptk}[1]{\upscr{\alpha_{#1}}{opt,k}}
\DeclareMathOperator*{\argmax}{arg\,max}

\def\real{\mathbb{R}}

\title{\LARGE \bf
Submodular Maximization with Limited Function Access
}

\author{Andrew Downie, Bahman Gharesifard, and Stephen L. Smith% <-this % stops a space
\thanks{Andrew Downie and Stephen L. Smith are with are with the Electrical and Computer Engineering at the University of Waterloo, Waterloo, ON, Canada
        {\tt\small \{adownie,stephen.smith\}@uwaterloo.ca}}%
\thanks{Bahman Gharesifard is with the Department of Electrical and Computer Engineering at the University of California, Los Angeles        {\tt\small gharesifard@ucla.edu}}%
}

\begin{document}
\maketitle

%%%%%%%%%%%%%%%%%%%%%%%%%%%%%%%%%%%%%%%%%%%%%%%%%%%%%%%%%%%%%%%%%%%%%%%%%%%%%%%%

\begin{abstract} 
We consider a class of submodular maximization problems in which decision-makers have limited access to the objective function. We explore scenarios where the decision-maker can observe only pairwise information, i.e., can evaluate the objective function on sets of size two. We begin with a negative result that no algorithm using only $k$-wise information can guarantee performance better than $k/n$. We present two algorithms that utilize only pairwise information about the function and characterize their performance relative to the optimal, which depends on the curvature of the submodular function. Additionally, if the submodular function possess a property called supermodularity of conditioning, then we can provide a method to bound the performance based purely on pairwise information. The proposed algorithms offer significant computational speedups over a traditional greedy strategy. A by-product of our study is the introduction of two new notions of curvature, the $k$-Marginal Curvature and the $k$-Cardinality Curvature. Finally, we present experiments highlighting the performance of our proposed algorithms in terms of approximation and time complexity.
\end{abstract}

%%%%%%%%%%%%%%%%%%%%%%%%%%%%%%%%%%%%%%%%%%%%%%%%%%%%%%%%%%%%%%%%%%%%%%%%%%%%%%%%
\section{Introduction}

Submodular maximization has recently generated interest in many decision-making problems, as it can provide strong performance guarantees for computationally difficult problems. Submodular functions are set functions that exhibit the property of diminishing returns. Submodular optimization is a well-studied subject, as these functions model many real-world problems in controls~\cite{submodularvoltagecontrol,Submodular_Power_Storage}, robotics \cite{UAV_Trajectory,SubmodularInformationGathering}, data processing \cite{DataSetSelection,MutliDocumentSummarization} and machine learning~\cite{ProbabilisticSubmodularLinearTime,EntropyGreedy} 

One practical difficulty in implementing algorithms for submodular maximization in complex settings is that the required function evaluations are computationally expensive. This can be attributed to the large-scale characteristics of the system~\cite{SubmodularStreaming}, application-specific constraints such as communication constraints~\cite{distsubmaxpartition}, or the type of data the objective function is evaluating~\cite{SensorPlacementWaterNetworks}. In its most common form, the submodular function is treated as a \textit{value oracle}, which is repeatedly queried by a greedy strategy to maximize the objective function. Therefore, it is inherently assumed that one can evaluate the function for sets of any size. In practice, however, it may only be possible to evaluate the functions on smaller set sizes due to computation cost or limitations imposed. Consider the setting where a company is selecting locations for several new retail stores. The total revenue received by a set of store locations can be modelled as a submodular function: As more stores are added, the marginal benefit of adding a new store is reduced. In the classical greedy algorithm for submodular maximization, we assume we have access to a value oracle to evaluate subsets of store locations. Armed with this oracle, we iteratively add a new store $s_k$ to the existing set $\{s_1,\dots,s_{k-1}\}$ by selecting the location $s_k$ that maximizes the marginal benefit $f(s_1,\dots,s_{k-1},s) - f(s_1,\dots,s_{k-1})$. To evaluate this quantity, the oracle must accurately model the revenue of $k$ stores, which can be challenging in practice due to their complex interactions: for example, $s_k$ may reduce the revenue at some $s_i$, which then may affect some other store's revenue. 

Motivated by the lack of access to the full value oracle in practical settings, in this paper, we seek to determine how well we can approximate the maximum value of a submodular function when we have access to a limited set of function values. For most parts of this paper, we focus on the case where we can access function values for single elements $f(s_i)$ and for pairs of elements $f(s_i,s_j)$. We refer to this as \emph{pairwise information}. In the motivating example, this corresponds to knowing the total revenue for a single store and the total revenue for any two stores together and nothing more. Note that this restriction on information is severe. A submodular function on a base set of $N$ elements can be represented as a look-up table with $2^N$ values. If only singleton and pairwise information are available, this means we have access to only $N(N+1)/2$ values. While we focus on pairwise information, we also extend most results to the case of $k$-wise information, where we can evaluate any set of size at most $k$.

\textit{Statement of Contributions:}
We consider the submodular maximization problem where information on the underlying function is limited, in that we only have access to evaluations of sets of size at most $k$. Let $X$ be the base set of elements we are optimizing over and $n$ be the maximum number of elements that can be in our solution set. We begin with a negative result; namely, there exists a class submodular functions for which no algorithm subject to this information constraint can guarantee performance better than $k/n$ of optimal. In light of this, we propose a class of functions where we can upper and lower bound the marginal gains the objective function in terms of pairwise information. Using these bounds, we propose two simple greedy algorithms that utilize only pairwise information. We introduce two new notions of curvature named the $k$-Marginal Curvature and the $k$-Cardinality Curvature, which capture ``how submodular'' a function is. We then adapt a previous result for approximate value oracles to prove performance bounds for the two algorithms in terms of our new notions of curvature. The two notions provide a new way to understand submodular functions and may be of independent interest. We also show that, using only pairwise information and an additional assumption called supermodularity of conditioning on the function, we can compute optimality bounds for a given solution. We illustrate that the structure of the lower and upper bound estimates of the marginal gains can be exploited to produce an algorithm that runs in exactly $\mathcal{O}(|X|\cdot n)$ without the assumption of a value oracle. Finally, we show experimental results for an autonomous ride service coverage problem that highlights the effectiveness of the algorithms and the time complexity advantages.

\textit{Related Work:}
It is well known that maximizing a submodular function subject to a cardinality constraint is NP-hard~\cite{Nemhauser}, but if the function is normalized and monotone, then a greedy algorithm provides an approximation factor of $(1-1/e)$. This paper considers a similar problem but with additional information constraints. Some important general submodular functions where the greedy algorithm has been extensively explored are graph cut~\cite{MutliDocumentSummarization}, mutual information~\cite{SubmodularSurvey}, set cover~\cite{SubmodularSetCover} and facility location~\cite{SubmodularSurvey}. Other traditional constraints that have been considered for submodular maximization include knapsack \cite{SubmodularKnapsackConstraints}, budget \cite{MutliDocumentSummarization} and matroid constraints \cite{MatroidConstriant,Robust_correlated}. Submodular maximization has also been studied in the context of robotics and controls, being used in applications such as sensor coverage~\cite{SubmodularSurvey}, sensor selection for Kalman filtering~\cite{weak_submodular,Kalman_sensor,submod_sensor_scheduling}, multi-robot exploration objectives~\cite{multi_robot}, voltage control~\cite{submodularvoltagecontrol}, multi-agent target tracking~\cite{resilient_target_tracking}, informative path planning~\cite{inform_path_planning}, and control input selection~\cite{input_selection}.

Recently, other information constraints are being considered in the context of distributed submodular maximization. In these scenarios, a team of agents are attempting to maximize a submodular objective function collaboratively. Each agent has access to their own set of actions and can observe a limited number of decisions made by other agents~\cite{LimitedInformation,distsubmaxpartition,ImpactOfInfromation,MessagePassing,ParrellelExecution}. In contrast, we consider the case where each decision-maker has limited access to the function $f$ itself. 

Another related concept is the idea of an \textit{approximate value oracle}~\cite{revisiting}. This refers to a black box that takes as input a set (or set and a new element) and outputs an approximation to the true function value (or an approximation to the marginal gains). Results can then be derived on the quality of the resulting solutions for a greedy algorithm using the approximate value oracle. The performance of these algorithms is a function of the approximation factor for the approximate value oracle. In this paper, we extend these approximation oracle results to provide approximation bounds in terms of new notions of curvature.

Finally, another aspect of submodular maximization is the computational efficiency of the greedy strategy~\cite{lazierthanlazy,SubmodularStreaming,ScalingViaPrunedSubmodularityGraphs,DistributedSubmodularMaximization}. Under the assumption that the strategy has access to a value oracle for the objective that can be computed in constant time, the time complexity is $\mathcal{O}(|X|\cdot n)$~\cite{lazierthanlazy}. Even though the time complexity is polynomial, for $X$ with large cardinality, the greedy strategy can become prohibitively expensive to execute. Consequently, alternative implementations have been provided for the greedy strategy that improve computational efficiency by leveraging streaming techniques to only pass over the set $X$ once~\cite{SubmodularStreaming,StreamingPlusPlus}, or by realizing the set $X$ as a tree and pruning nodes and edges to reduce the size of $X$~\cite{ScalingViaPrunedSubmodularityGraphs}. Parallelized implementations of the greedy strategy and strategies that greedily maximize over randomly sampled subsets of $X$~\cite{lazierthanlazy} have addressed computational issues for large problems. Most of these techniques focus on reducing the search space being optimized over but do not address the cost of computing the objective function. 

\section{Problem Definition and Inapproximability}
Let $X$ be a set of elements and $2^X$ be the power set of those elements. A set function $f:2^X\rightarrow \real_{\geq0}$ is submodular if the following property of diminishing returns holds: For all $A \subseteq B \subseteq X$ and $x\in X\backslash B$ we have $$f(A\cup\{x\})- f(A) \geq f(B\cup\{x\})- f(B).$$
We refer to $f(A\cup\{x\})- f(A)$ as the marginal return of $x$ given $A$, denoted by $f(x|A)$.  For simplicity, we denote the objective value of a singleton $f(\{x\})$ by $f(x)$. We also denote marginal return of $x$ with respect to a singleton set $A = \{y\}$ by $f(x|y)$ and refer to it as the pairwise marginal return of $x$ given $y$. In addition to submodularity, throughout this paper, we assume that the functions satisfy 
\begin{enumerate}
    \item Monotonicity: For all $A\subseteq B \subseteq X$, $f(A) \leq f(B)$,
    \item Normalization: $f(\emptyset) = 0$.
\end{enumerate}

Another property of submodular functions we utilize is the notion of curvature~ \cite{Matriod_curvature,fast_multi_stage}. The curvature of a submodular function $f$ is defined as 
\begin{equation}
    c = 1-\min_{A\subseteq X, x\in X\backslash A}\frac{f(x|A)}{f(x)}.\label{traditional_curvature}
\end{equation}
Note that if the value of $c = 0 $, the function is modular. 

We start by recalling the problem of maximizing a submodular function over the uniform matriod. Let $X$ be a set of elements and let $f:2^{X}\rightarrow \real_{\geq0}$ be a monotone normalized submodular function. We wish to solve the following problem:
\begin{align}
   &\max_{S \subseteq X}f(S)\label{basicproblem}\\
   &\text{s.t. } |S|\leq n \nonumber
\end{align}
This is the classical submodular maximization problem that can be solved to an approximation factor of $(1-1/e)$ using the following simple greedy algorithm, see \cite{Nemhauser}:
\begin{align}
    x_i = \argmax_{x\in X\backslash S_{i-1}}f(x|S_{i-1})\\
    S_i = S_{i-1}\cup\{x_i\}\nonumber,
\end{align}
where $i$ is the iteration of the algorithm and $S_i$ is the solution produced after $i$ iterations. In what follows, we often refer to this strategy as the full information greedy algorithm. 
A key focus of our contributions is to understand the limitations of algorithms that only have access to partial information about the objective function. We make this precise in the next definition. 
\begin{definition}($k$-wise Information)
Given a submodular function $f$, the $k$-wise information set is defined as the set of tuples $\{(S,f(S)) | S\subseteq X, |S|\leq k\}$. When $ k=2$, we refer to this as pairwise information. 
\end{definition}
An algorithm that has access to $k$-wise information can only use evaluations of $f$ on sets of size $k$ to form a decision. We denote the class of such algorithms by $\subscr{\Pi}{k-wise} $, or $\subscr{\Pi}{pairwise} $ when~$k=2$.
The main objective that we have in mind is to study Problem~\ref{basicproblem} with such limitations.
We now present a negative result that addresses the inapproximability of this problem. 
\begin{proposition}
\label{prop:inapprox}
  Consider Problem~\ref{basicproblem} with $k$-wise information. Then for every algorithm $\pi\in \subscr{\Pi}{k-wise}$, there exists a submodular function $f$ such that
  \[
  f(S^{\pi}) \leq \frac{k}{n} f(S^*),
  \]
  where $S^{\pi}$ is the solution constructed by $\pi$ and $S^*$ is the optimal solution.
\end{proposition}
\begin{proof}
  We begin by constructing a normalized, monotone submodular function $f$.  Consider a set $X$ that is partitioned into two disjoint sets $X = V \cup V^*$, where $|V^*| = n$ and~$|V| \geq n$. We define the function $f: 2^X \to \mathbb{R}_{\geq 0}$ as:
  \[
  f(S) = \min\{|S\cap V|,k\} + |S \cap V^*|.
  \]
  This function is normalized and monotone, and given $k$, it assigns a value of $k$ to all sets of size $k$. The $V$ can be thought of as the general set and $V^*$ is a special set where you are guaranteed to get value if you selected an element from $V^*$. The function, counts the number of elements of $S$ that are in~$V^*$.  However, for all sets $S$ where $|S| \leq k$ get mapped to their cardinality. 
  
    We now show that $f$ is also submodular.  Consider any two sets $A \subset B \subset X$ and an element $x \in X \setminus B$.  We show that 
   \[
   f(x|A) \geq f(x|B).
   \]
   First notice that $f(x|A)$ and $f(x|B)$ are each either $0$ or $1$, since adding an element can increase the function value by at most one.  There are two cases to consider:

   Case 1 ($x\in V^*$):  In this case $f(x|A) = 1$, since $A\cup \{x\}$ has one more element in $V^*$ than $A$.  Since $f(x|B) \leq 1$, the result follows.

   Case 2  ($x\in V$):  We assume that $f(x|B) = 1$, as otherwise the result holds.  Since $f(x|B) = 1$ and $x\in V$, we must have $|B \cap V| < k$.  But this implies that $|A \cap V| < k$ since $A \subseteq B$.  Thus $f(x|A) = 1$ and the result holds.

   For any set $S$ with $|S| \leq k$  we have that $f(S) = |S|$ which reveals no information on which elements of $S$ are in $V$ or $V^*$.  Hence, from the perspective of an algorithm in  $\subscr{\Pi}{k-wise}$, the elements in $X$ are indistinguishable.  Given any algorithm $\pi \in \subscr{\Pi}{k-wise}$, there exists an assignment of the elements of $X$ to $V$ and $V^*$ such that $f(S^{\pi}) = k$.  Since the optimal solution is $S^* = V^*$ achieving a value of $f(S^*) = n$, we obtain the desired result.
\end{proof}

This result highlights the challenges that arise under $k$-wise information constraints. As shown in the proof, there exists functions where their marginal returns with respect to sets of size $k$ or less, tell you nothing about the marginals with respect to sets with sizes greater than $k$. As we see in the sensor coverage example depicted in Figure \ref{fig:choices}, the marginal returns of a single sensor given subset of other sensors can be approximated by taking it's area and subtracting the pairwise overlaps of it self to the rest of the element in the subset. These overlaps can be derived from utilizing only pairwise marginals, therefore pairwise information provides us with information about the higher order marginals.

In the following two sections, we characterize functions where the marginals with respect to sets of size $k$ or smaller are informative of the higher order marginals using new notions of curvature.

\begin{figure}
    \centering
    \includegraphics[width = 0.4\textwidth]{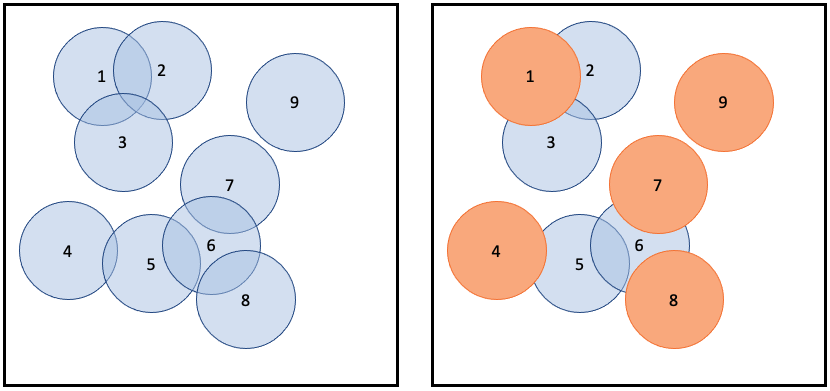}
    \caption{Left: Set of sensor footprints that an algorithm could potentially select. The objective is to select five sensors that maximize the area covered by the union of their footprints. Right: A set of 5 sensors that maximize the desired objective function. See that in this particular example the sensors that maximize the area covered have the minimum pairwise overlaps between them.}
    \label{fig:choices}
\end{figure}

\section{Pairwise Algorithms}
Our main objective in what follows is to leverage pairwise information to find an approximate solution to Problem~\ref{basicproblem}.
\subsection{Optimistic Algorithm}
 A natural strategy is to greedily select elements that maximize the \textit{estimated} marginal return using only pairwise information. First, note that 

\begin{equation}
    \min_{x_j\in A}f(x|x_j)\geq f(x|A), \label{ub}
\end{equation}
which holds by submodularity of $f$, because for all $\{x_j\} \subseteq A$, we have $f(x|x_j) \geq f(x|A)$. 
We will define a simple estimate of the marginal returns of $f$ as the left hand side of~\eqref{ub}$$\bar{f}(x|A) := \min_{x_j\in A}f(x|x_j).$$ The pairwise marginal for all $x_j\in A$ upper bounds $f(x|A)$ and hence we choose the minimum as it is the best available estimate of the true value of $f(x|A)$. In a nearly identical style to the classical greedy strategy, we now define an algorithm as follows:
\begin{align}
    x_i = \argmax_{x\in X\backslash S_{i-1}}\bar{f}(x|S_{i-1})\label{opt_greedy}\\
    S_i = S_{i-1}\cup \{x_i\}\nonumber.
\end{align}
Throughout this paper, we will refer to~\eqref{opt_greedy} as the \textit{optimistic} algorithm. In essence, the optimistic algorithm aims to greedily select elements with maximum \emph{potential} marginal return.

\subsection{Approximate Value Oracles}
To characterize the performance of the optimistic algorithm given by~\eqref{opt_greedy}, we consider the problem through the lens of maximizing a submodular objective function via surrogate objective functions. Following~\cite{fast_multi_stage}, we will discuss how to determine performance guarantees when using such surrogates. 

Let $\{x_1,\dots,x_n\} \subseteq X$ be the choices made by some algorithm. We denote by $S_i~=~\{x_1,\dots,x_i\}$ the choices selected after the $i$th iteration. Now let $\{x_1^g,\dots,x_n^g\} \subseteq X$ be such that each $x_i^g$ maximizes the marginal return of $f$ conditioned on $S_{i-1}$, i.e., 
$$x_i^g = \argmax_{x\in X\backslash S_{i-1}}f(x|S_{i-1}).$$

The set $\{x_1^g, \dots, x_n^g\}$ represents the elements that a greedy algorithm with full information about the objective $f$ would have selected if it had previously selected $S_{i-1}$.
Using these values, we can now measure the quality of a given algorithm's choices compared to that of an algorithm with full information about the objective. We do this by finding  $\alpha_i \in \real_{+}$, for $i~\in~\{1,\dots,n\}$ such that
\begin{equation}
    \alpha_i f(x_i|S_{i-1})\geq f(x_i^g|S_{i-1}).\label{approx_factor_def}
\end{equation}
By the greedy choice property of $x_i^g$, we have that 
\[
f(x_i|S_{i-1})\leq f(x_i^g|S_{i-1}).
\]
Hence, $\alpha_i \geq 1$, for all $i\in\{1,\dots,n\}$. From this point on, we call each $\alpha_i$ the approximation factor associated with $x_i$. 

In the general framework proposed in~\cite{fast_multi_stage}, the objective is to greedily maximize multiple surrogate objective functions, and to use these to generate approximate solutions. For our problem of submodular maximization with only pairwise information, we simply maximize using a single surrogate function $\bar{f}(x|S)$. We provide a simplified version of~\cite[Theorem~1]{fast_multi_stage} as follows.
\begin{theorem}
\label{approx_theorem}
Suppose that $S = \{x_1,\dots,x_n\}\subseteq X$ is the set of elements selected by an algorithm and $\{\alpha_1,\dots,\alpha_n\}$ are the set of approximation factors that satisfy~\eqref{approx_factor_def}. Let $S^*$ be the optimal solution to Problem~\ref{basicproblem}. Then
\begin{equation}
    f(S) \geq \left(1-e^{-\frac{1}{n}\sum_{i = 1}^n\frac{1}{\alpha_i}}\right)f(S^*). \label{eq:approx_theorem}
\end{equation}
\end{theorem}
Given that we only maximize one surrogate function and in order to keep this paper self-contained, we provide a proof of this result in the Appendix which is simpler than the general result established in~\cite{fast_multi_stage}. Note that Theorem~\ref{approx_theorem} relies on $f$ being a normalized monotone submodular function. This result can be applied to any algorithm for Problem~\ref{basicproblem}, not just algorithms that only have access to pairwise information. An interesting remark about Theorem~\ref{approx_theorem} is that the performance bound depends essentially on the average of the approximation factors. Some of these factors could be large compared to the others, but as long as most of them are small, good performance is maintained.

\subsection{Optimistic Algorithm Approximation Performance}
\label{optimimistic_sub_section}
We aim to provide approximation guarantees for the optimistic algorithm. To give an intuition for what we are about to present, we consider the following example.
\begin{exmp}
\label{pile_example}
Consider the scenario depicted in Figure~\ref{fig:example_optimisitic_better}. Here, we wish to select four sensors to maximize the area of their combined footprints.  One of the simplest algorithms that satisfies the pairwise information constraint is the \emph{uninformed} greedy strategy \begin{align}
    x_i &= \argmax_{x\in X\backslash S_{i-1}}f(x) \label{uninformed_greedy}\\
    S_i &= S_{i-1}\cup \{x_i\}. \nonumber
\end{align}
We refer to this algorithm as \emph{uninformed} because it only use the most basic information about $f$ which it's values evaluated on single elements. For the \emph{uninformed} greedy strategy, the scenario described in Figure \ref{fig:example_optimisitic_better} could potentially lead to poor performance. This strategy cannot distinguish between its choices and therefore could select four sensors that almost perfectly overlap with each other (i.e., in the same pile), resulting in a low objective value. Alternatively, if we had used the optimistic algorithm, once one element is selected from a pile, the pairwise upper bound on the other elements in a pile would be low. In later iterations, the optimistic algorithm would avoid selecting elements in piles where elements have been previously selected from. Interestingly, we see that for each $i$, the difference between $f(x_i|S_{i-1})$ and $\bar{f}(x_i|S_{i-1})$ is small. We notice that in these scenarios, the value of $\bar{f}(x_i|S_{i-1})$ provides accurate information about the value of $f(x_i|S_{i-1})$. This is the idea that we want to capture in the following result. 
\begin{figure}
    \centering
    \includegraphics[width = 0.3\textwidth]{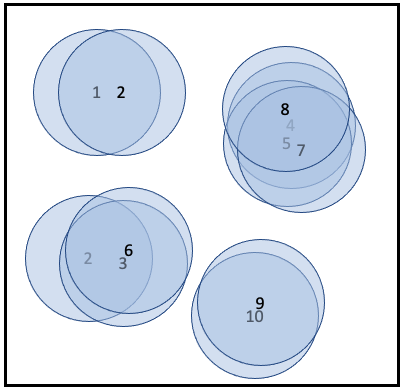}
    \caption{Example sensor coverage configuration where the optimistic algorithm performs better than uninformed greedy strategy}
    \label{fig:example_optimisitic_better}
\end{figure}

\end{exmp}

\begin{theorem}
\label{ub_approx_theorem_}
Let $S_{i-1}\subseteq X$ be the partial solution of optimistic algorithm after $(i-1)$ iterations, and let $x_i\in X$ be the element selected at the $i$th iteration. Then we have that,
\begin{equation}
\alphaopt{i}=\begin{cases} 
      1 & i\in\{1,2\}\\
      \frac{\bar{f}(x_i|S_{i-1})}{f(x_i|S_{i-1})}& i > 2 \label{eq:ub_factor_}\\
   \end{cases}
\end{equation}
satisfy \eqref{approx_factor_def} for all $i\leq n$.
\end{theorem}

\begin{proof}
Let $x_i^g$ be the true greedy choice at iteration $i$ given $S_{i-1}$. 
For $i = \{1,2\}$, we have that 
\[
\bar{f}(x|S_{i-1}) = f(x|S_{i-1}).
\]
Hence,  $x_i = x_i^g$ and therefore, we can let $\alphaopt{1} = \alphaopt{2} = 1$. For $i>2$, based on from \eqref{approx_factor_def} let $\alpha_i^{\min}$ be the smallest value such that \eqref{approx_factor_def} which can be written as
\begin{equation}
    \alpha_i^{\min}= \frac{f(x_i^g|S_{i-1})}{f(x_i|S_{i-1})}\label{min_def}.
\end{equation}
Any approximation factor $\alpha_i$ such that $\alpha_i\geq \alpha_i^{\min}$ will satisfy~\eqref{approx_factor_def}. We will now upper bound $\alpha_i^{\min}$ as follows:
\begin{align}
    \alpha_i^{\min} &= \frac{f(x_i^g|S_{i-1})}{f(x_i|S_{i-1})}\nonumber\\
    &\leq \frac{\bar{f}(x_i^g|S_{i-1})}{f(x_i|S_{i-1})}\label{by_ub_def_}\\
    &\leq \frac{\bar{f}(x_i|S_{i-1})}{f(x_i|S_{i-1})}\label{by_greedy_choice_}
\end{align} 
where~\eqref{by_ub_def_} holds by definition and~\eqref{by_greedy_choice_} holds by the greedy choice property of the optimistic algorithm. Setting $\alphaopt{i}$ to be the right hand side of~\eqref{by_greedy_choice_}, we conclude the proof.
\end{proof}

The following corollary is an immediate consequence of Theorem \ref{ub_approx_theorem_}.

\begin{corollary}
\label{ub_corollary_}
Let $S\subseteq X$ be the solution produced by the optimistic algorithm and $S_{i-1}\subseteq S$ be the partial solution after $(i-1)$ iterations of the optimistic algorithm and let $x_i\in S$ be the element selected at the $i$-th iteration, then we have
\begin{equation}
    f(S) \geq \left(1- e^{-\frac{1}{n}\left(2 + 
    \sum_{i= 3}^n \frac{f(x_i|S_{i-1})}{\bar{f}(x_i|S_{i-1})}\right)}\right)f(S^*)\label{ratio_bound}.
\end{equation}
\end{corollary}

We see that the approximation performance of the algorithm is dictated by the sum in the exponent. We can interpret the exponent as the mean of the set $$\left\{1,1,\frac{f(x_3|S_{2})}{\bar{f}(x_3|S_{2})}, \dots, \frac{f(x_n|S_{n-1})}{\bar{f}(x_n|S_{n-1})}\right\}.$$ That implies that, to get adequate performance from the optimistic algorithm, we need the value of $\bar{f}(x_i|S_{i-1})$ to be close to $f(x_{i}|S_{i-1})$ on \textit{average}.

We also see that the $\frac{f(x_{i}|S_{i-1})}{\bar{f}(x_{i}|S_{i-1})}$ is closely related to the traditional notion of curvature. Let us define the following quantity.

\begin{definition}[$k$-Marginal Curvature]\label{def:k-marginal}
The $k$-marginal curvature of $f$ given $S\subseteq X$ and $x\in X\backslash S$ is defined as 
\begin{equation}
    c_k(x|S) = 1- \max_{A \subseteq S, |A| < k}\frac{f(x|S)}{f(x|A)}.
\end{equation}
\end{definition}
To analyze the optimistic algorithm that only has access to pairwise information, we will work with the $2$-marginal curvature, which can be written as 
$$c_2(x|S) = 1- \frac{f(x|S)}{\bar{f}(x|S)}.$$

\begin{remark}
This $2$-marginal curvature is similar to the traditional notion of curvature \eqref{traditional_curvature}, but characterizes the relationship between the values of the pairwise upper bounds $\bar{f}(x|S)$ and true values of $f(x|S)$. A key difference between the two curvatures, is that there exist functions where the values of the $2$-marginal curvatures can be close to $0$ even though the value of traditional curvature is close to $1$. The sensor coverage function, described in Figure \ref{fig:example_optimisitic_better}, is an example of a function where the traditional curvature is close to $1$ and the  values of the $2$-marginal curvatures are close to $1$.
\end{remark}
This allows us to rewrite~\eqref{ratio_bound} as follows.
\begin{equation}
    f(S) \geq \left(1- e^{-\frac{1}{n}\left(2 + 
    \sum_{i= 3}^n 1- c_2(x_i|S_{i-1})\right)}\right)f(S^*) \label{opt_curve_corollary}
\end{equation}
We now can characterize the worst case performance in terms of the average of the 2-marginal curvatures, which capture the intuition from Example~\ref{pile_example}. In Figure~\ref{fig:example_optimisitic_better} the elements will have $2$-marginal curvatures close to zero, resulting in a strong approximation bound.

\subsection{Extension to $k$-wise Information}
The analysis from Subsection~\ref{optimimistic_sub_section} can be naturally extended to the problem of observing $k$-wise information. Suppose that we wish to approximately solve Problem \ref{basicproblem} using an algorithm that only has $k$-wise information available to it. We extend the pairwise optimistic algorithm to the $k$-wise optimistic algorithm as follows. Let us define an upper bound on the marginal returns using $k$-wise information. Let $x\in X$ and $S\subseteq X$ then we have the following upper bound
\begin{equation}
    \min_{A\subseteq S,|A|<k}f(x|A) \geq f(x|S),\label{k-wise upperbound}
\end{equation}
which holds by submodularity of $f$. We will denote the left hand side of \eqref{k-wise upperbound} as $$\bar{f}_k(x|S) := \min_{A\subseteq S,|A|<k}f(x|A).$$ We can now define the $k$-wise optimistic algorithm as 
\begin{align}
    x_i = \argmax_{x\in X\backslash S_{i-1}}\bar{f}_k(x|S_{i-1})\label{k-opt_greedy}\\
    S_i = S_{i-1}\cup \{x_i\}\nonumber.
\end{align}
By submodularity, we have that $\bar{f}(x|S) \geq \bar{f}_k(x|S)\geq f(x|S)$ for $k>1$. 
Our result is stated next. 
\begin{theorem}
\label{k_wise_ub_approx_theorem}
Let $S_{i-1}\subseteq X$ be the partial solution of $k$-wise optimistic algorithm after $(i-1)$ iterations and $x_i\in X$ be the element selected during the $i$th iteration. Then we have that
\begin{equation}
\alphaoptk{i}=\begin{cases} 
      1 & i \leq k\\
      \frac{\bar{f}_k(x_i|S_{i-1}))}{f(x_i|S_{i-1})}& i > k \label{eq:k_wise_ub_factor}\\
   \end{cases}
\end{equation}
satisfy \eqref{approx_factor_def} for all $i\leq n$.
\end{theorem}

The proof is nearly identical to the one of Theorem~\ref{ub_approx_theorem_} and can be found in the appendix. The following result is an immediate consequence of Theorem~\ref{k_wise_ub_approx_theorem}. 
\begin{corollary}
\label{k_corollary}
Let $S\subseteq X$ be the solution produced by the $k$-wise optimistic algorithm, then we have

\begin{equation}
    f(S) \geq \left(1- e^{-\frac{1}{n}\left(k+ \sum_{i=k+1}^n\frac{f(x_i|S_{i-1})}{\bar{f}_k(x_i|S_{i-1})}\right)}\right)f(S^*)\label{eq:k_corollary}.
\end{equation}
\end{corollary}

Note that using this result, we can rewrite~\eqref{eq:k_corollary} as
\begin{equation}
    f(S) \geq \left(1- e^{-\frac{1}{n}\left(k + 
    \sum_{i= k+1}^n 1- c_k(x_i|S_{i-1})\right)}\right)f(S^*).
\end{equation}
Comparing this to the scenario with pairwise information, we see that access to more information improves approximation guarantees. In particular, since $c_2(x|S)\geq c_{k}(x|S)$ for all $S\subseteq X$ and $x \in X\backslash S$, we can guarantee that 
\begin{align}
    \frac{1}{n}
    &\left(k + 
    \sum_{i= k+1}^n 1- c_k(x_i|S_{i-1}) \right) \geq\nonumber\\
    &\frac{1}{n}\left(2 + \sum_{i= 3}^n 1- c_2(x_i|S_{i-1})\right).
\end{align}
This implies that approximation bound in Corollary~\ref{k_corollary} is stronger than Corollary~\ref{ub_corollary_}.

Having access to $k$-wise information provides us with stronger approximation bounds, but we trade off computation performance. We are required to compute the minimum marginal overall subsets $A\subseteq S_{i-1}$, where $|A|< k$ for each $x\in X$. When $k\leq |S_i|$, we need to check $|S_i|\choose{k-1}$ subsets of $S$ to find the minimum. This becomes expensive to do computationally as $S_{i-1}$ grows larger. If $k~=~3$, the computation of each marginal is quadratic in $|S|$ and can be expensive to compute. From a practical perspective, we can actually compute the pairwise optimistic algorithm efficiently; we will discuss this in Section~\ref{time_complexity}. 

\section{Pairwise Algorithms Utilizing Supermodularity of Conditioning}

In this section we introduce an additional property that a submodular function can possess which is use useful when we only have access to pairwise information. This property which is called supermodularity of conditioning and is related to montonicity, allows us to compute performance bounds for an algorithm ``post-hoc'' using only pairwise information.

\subsection{Post-Hoc Performance Bounds}
To characterize the approximation performance of an algorithm $\pi\in \subscr{\Pi}{pairwise}$ using Theorem \ref{ub_approx_theorem_} we are required to compute the full marginal of the function $f$, which may not be available in practice. Alternatively, after we execute an algorithm $\pi$ to produce a solution 
\[
S^{\pi} = \{x_1^\pi,\dots,x_n^\pi\}\subseteq{X},
\]
we can determine $\gamma\in \real_{\geq 0}$ such that 
\begin{equation}\label{eq:gamma}
    f(S^{\pi})\geq \gamma f(S^*).
\end{equation}  
This is done by bounding $\alpha_i$ in~\eqref{approx_factor_def} using only pairwise information. As described in the proof of Theorem~\ref{ub_approx_theorem_}, the smallest value of $\alpha_i$ that will satisfy~\eqref{approx_factor_def} is $\alpha_i^{\min}$. 

Let  
\[
S^\pi_{i} = \{x_1^\pi,\dots,x_i^{\pi}\}\subseteq{X}
\]
be the partial solution of $S^{\pi}$. Then, we have that 
$$\alpha_i^{\min}= \frac{f(x_i^g|S^{\pi}_{i-1})}{f(x_i^{\pi}|S^{\pi}_{i-1})}\leq \frac{\max_{x\in X\backslash S^{\pi}_{i-1}}\bar{f}(x|S^{\pi}_{i-1})}{f(x_i^{\pi}|S^{\pi}_{i-1})}.$$
By lower bounding $f(x_i^{\pi}|S^{\pi}_{i-1})$ using pairwise information, we obtain an $\alpha_i$ that satisfies~\eqref{approx_factor_def}, and therefore, Theorem \ref{approx_theorem} allows us to find $\gamma$ that 
satisfies~\eqref{eq:gamma}.

If we impose an additional monotonicity property on $f$ called \textit{supermodularity of conditioning}, then we are able find a lower bound on the marginal returns of $f$ using only pairwise information.

\begin{definition}(Supermodularity of Conditioning)
A submodular function $f$ possess the property of \textit{supermodularity of conditioning} if for all $S\subseteq X$, $A\subseteq B\subseteq X$ and $C \subseteq X\backslash B$, we have that \begin{equation}
    f(S|A)-f(S|A,C) \geq f(S|B)-f(S|B,C). \label{super_condition}
\end{equation}
\end{definition}

Supermodularity of conditioning is a higher order monotonicity property which describes how the \textit{redundancy} of two sets are affected by conditioning. Suppose that $A = \emptyset$, the redundancy between $S$ and $C$ is $f(S)-f(S|C)$, then by further conditioning by $B$ reduces the redundancy. Supermodularity of conditioning has been used in the context of distributed
submodular maximization in~\cite{distsubmaxpartition}. Some notable examples of functions that exhibit supermodularity of conditioning are weighted set coverage, area coverage and probabilistic set coverage. We recall the following result from~\cite{distsubmaxpartition}. 
\begin{lemma}{(Pairwise Redundancy Bound)}
\label{lemma:pairwise_redundancy_bound}
Let $f$ be a submodular function on $X$ that exhibits supermodularity of conditioning and let $A,B,C \subseteq X$ be disjoint subsets. Then
\begin{equation}
    f(A|B)-f(A|B,C) \leq \sum_{c\in C}f(c)-f(c|A)\label{reduncany_bound}.
\end{equation}
\end{lemma} 
We can now state a result establishing a lower bound on the marginal return.
\begin{theorem}(Pairwise Marginal Lower Bound)
\label{theorem:pairwise_lower_bound}
Let $f$ be a submodular function that exhibits supermodularity of conditioning. Then for $x\in X$ and $S\subseteq X$
\begin{equation}
    f(x|S)\geq f(x)-\sum_{x_j\in S}f(x)-f(x|x_j).
    \label{lb}
\end{equation}
\end{theorem}
\begin{proof}
%Let $A = \{x\}$, $B = \emptyset$ and $C = S$.
Since $f$ exhibits supermodularity of conditioning, applying Lemma \ref{lemma:pairwise_redundancy_bound} with $A = \{x\}$, $B = \emptyset$ and $C = S$, we have
\begin{align}
    f(x)-f(x|S) &\leq \sum_{x_j\in S}f(x_j) - f(x_j|x)\nonumber\\ 
    &= \sum_{x_j\in S}f(x) -f(x|x_j) \label{theorem:pairwise_lower_bound:eq:2},
\end{align}
where the last equality hold by the definition of the marginal return, yielding the result.
\end{proof}

For $x\in X$ and $S\subseteq X$ we define $$\underline{f}(x|S):= f(x)-\sum_{x_j\in S}f(x)-f(x|x_j).$$

We can now directly use this lower bound on the marginal returns to bound $\alpha_i^{\min}$. We now have that 
$\alpha_i^{\min} \leq \upscr{\alpha_i}{pairwise}$, where 
\begin{equation}
\upscr{\alpha_i}{pairwise}= \left\{ \begin{array}{cc}
        \frac{\max_{x\in X\backslash S^{\pi}_{i-1}}\bar{f}(x|S^{\pi}_{i-1})}{
        \underline{f}(x_i^{\pi}|S^{\pi}_{i-1})} & \underline{f}(x_i^{\pi}|S^{\pi}_{i-1})\geq 0 \\
         \infty& \underline{f}(x_i^{\pi}|S^{\pi}_{i-1})<  0 
    \end{array}\right.. \label{post_hoc_alpha}
\end{equation}

Now, $\upscr{\alpha_i}{pairwise}$ satisfies~\eqref{approx_factor_def} and is computable using only pairwise information. Note that we need to set $\upscr{\alpha_i}{pairwise}= \infty$ when $f(x_i^{\pi}|S^{\pi}_{i-1}) <0$ as otherwise, the resulting $\upscr{\alpha_i}{pairwise}$ would not upper bound $\alpha_i^{\min}$.

We now present an algorithm that, given $\pi\in \subscr{\Pi}{pairwise}$ and pairwise information about $f$, produces a worst-case performance bound $\gamma$ such that the solution $S^{\pi}$ satisfies $f(S^{\pi})~\geq \gamma f(S^*)$.
\begin{algorithm}
\SetAlgoLined 
\KwIn{$S^{\pi}$, $X$}
\KwResult{$\gamma$ such that $f(S)\geq \gamma f(S^*)$}
$S_0^{\pi} \leftarrow \emptyset$\;

\For{$i \leftarrow 1,\dots, n$}{
    select $x_i^{\pi}$ from $S^{\pi}\backslash S_{i-1}^{\pi}$\;
    \uIf{$\underline{f}(x_i^{\pi}|S_{i-1}^{\pi}) \geq 0$}{
        $\alpha_i \leftarrow \frac{\max_{x\in X\backslash S_{i-1}^{\pi}}\bar{f}(x|S_{i-1}^{\pi})}{\underline{f}(x_i^{\pi}|S_{i-1}^{\pi})}$\;
    }
    \uElse{
        $\alpha_i\leftarrow \infty$\;
    }
    $S_{i}^{\pi} \leftarrow S_{i-1}^{\pi} \cup \{x_i^{\pi}\}$\;
}
al$\gamma \leftarrow 1- e^{-\frac{1}{n}\sum_{i=1}^{n}\frac{1}{\alpha_i}}$\;
 \caption{Pairwise Information Post-Hoc Bound}
 \label{algo:posthoc}
\end{algorithm}

Algorithm~\ref{algo:posthoc} provides us the means to find performance bounds for an arbitrary algorithm $\pi$ given only pairwise information about $f$.  This does not guarantee the performance before execution, but it does provide a way to verify performance of an algorithm without having to explicitly compute $f(S^{\pi})$ or $f(S^*)$. 

We end this section with a few remarks about supermodularity of conditioning. 
\begin{remark}[On Supermodularity of Conditioning]
Note that assuming $f$ possess supermodularity of conditioning does not affect the hardness results for submodular maximization.  It is shown~\cite{distsubmaxpartition} that the weighted set cover problem possesses the supermodularity of conditioning while still satisfying the hardness results. 
\end{remark}

\subsection{Pessimistic Algorithm}
We next propose another pairwise algorithm which we will call the \textit{pessimistic} algorithm, given by
\begin{align}
    x_i = \argmax_{x\in X\backslash S_{i-1}}\underline{f}(x|S_{i-1})\label{pess_greedy}\\
    S_i = S_{i-1}\cup \{x_i\}\nonumber.
\end{align}
This algorithm enjoys similar guarantees as the optimistic algorithm when curvature assumptions are made but can outperform the optimistic algorithm in certain scenarios. Similar to the optimistic algorithm, we greedily select elements with the highest \textit{guaranteed value}, which is exactly greedily minimizing the approximation factor $\alpha_i$ in Algorithm \ref{algo:posthoc}. What differs from the optimistic algorithm is that we require the additional assumption of supermodularity of conditioning on the objective function, and take advantage of it. Later in our experimental results, we will show the effectiveness of the pessimistic algorithm for a probabilistic coverage problem.

Using an alternative definition of curvature, we can produce a similar performance bound as the optimistic algorithm. Let us define the  $k$-cardinality curvature as follows.
\begin{definition}[$k$-Cardinality Curvature]
\label{total_k_wise_curve_def}
Let $x\in X$ and $S\subseteq S$. We define the $k$-cardinality curvature $\tau_k$ as
 \begin{equation}
     \tau_k = 1-\min_{x \in X, A\subseteq X,|A|<k}\frac{f(x|A)}{f(x)}.
\end{equation}
\end{definition}
What differs between this notion of curvature and the $k$-marginal curvature is that it compares the values of the marginals of $x$ with respect to sets of size less than $k$ to the values of $f$ evaluated on a singletons.

The $2$-cardinality curvature can be written as
$$\tau_2 = 1-\min_{x \in X,y\in X\backslash \{x\}}\frac{f(x|y)}{f(x)}.$$
This quantity satisfies $$\tau_2 \geq 1-\frac{f(x|y)}{f(x)},$$ for all $x, y\in X$. 

Using the lower bound~\eqref{lb}, we have that
\begin{align}
    \underline{f}(x|S) &= f(x) -\sum_{x_j\in S}f(x)-f(x|x_j)\nonumber\\
    & = f(x)\left(1-\sum_{x_j\in S}1-\frac{f(x|x_j)}{f(x)}\right)\nonumber\\
    &\geq f(x)(1-|S|\tau_2).\nonumber
\end{align}
Since we know that $f(x|S) \geq 0$, 
\begin{equation}
   f(x|S)\geq  f(x)(1-\min\{|S|\tau_2,1\}). \label{lb_curve_lowerbound}
\end{equation}
This leads to the following result.

\begin{theorem}
\label{lb_approx_theorem}
Let $f$ be a normalized monotone submodular function that possesses supermodularity of conditioning. For the solution produced by the pessimistic algorithm the approximation factors,
\begin{equation}
    \alphapes{i}=\begin{cases} 
      1 & i \leq 2\\
      \frac{1}{1-\min\{(i-1)\tau_2,1\}}& i > 2 \label{alpha_pess}\\
  \end{cases}
\end{equation}
satisfy \eqref{approx_factor_def} for all $i\leq n$.
\end{theorem}
The proof can be found in the appendix. The following corollary immediately follows.
\begin{corollary}
\label{lb_corollary}
Let $S\subseteq X$ be the solution produced by the pessimistic algorithm then we have
\begin{equation}
    f(S) \geq \left(1- e^{-\frac{1}{n}\left(2 + \sum_{i=3}^n(1-\min\{(i-1)\tau_2,1\})\right)}\right)f(S^*)\label{pess_corollary}
\end{equation}
\end{corollary}

Similar to the optimistic algorithm, we see that if $\tau_2$ is small, $\underline{f}(x|S)$ closely represents the true value of $f(x|S)$. 
This bound on performance can be loose relative to the bound produced by Algorithm \ref{algo:posthoc} due to \eqref{lb_curve_lowerbound} being course. The post-hoc bound produced by Algorithm \ref{algo:posthoc} will provide a tighter bound on performance than Corollary \ref{lb_corollary}.

\begin{remark}
The new notions of curvature in  Definitions~\ref{def:k-marginal} and~\ref{total_k_wise_curve_def} are related to the traditional definition of curvature. Let $c$ be the traditional curvature as described in~\eqref{traditional_curvature}, $S\subseteq X$ and $x\in X\backslash S$, both the $k$-marginal and $k$-cardinality curvature have similar inequalities.
\begin{equation}
    c \geq c_k(x|S) \text{ and } c \geq \tau_k.
\end{equation}
There are also scenarios where $c$ can be $1$ and either $c_2(x|S)$ or $\tau_2$ can be small. In Figure \ref{fig:choices}, we see that $\tau_2$ will be small but $c(x|S)$ can be large. Suppose that the disks in Figure~\ref{fig:choices}, have area $1$, and let $x$ be disk 6 and $S$ be disk 5, 7 and 8. Then $c(x|S)$ will be large because, $\bar{f}(x|S) \approx 2/3$ and $f(x|S)$ is small resulting in a larger 2-marginal curvature. As previously described, $\tau_2$ is small in this example because for the disks, $x$ and $y$, that have the most overlap, we have $f(x|y) \approx 2/3$ and $f(x) = 1$. Figure \ref{fig:example_optimisitic_better}, describes the opposite case where $\tau_2\approx 1$ and  $c(x|S) \approx 0$ for any $S\subseteq X$ and $x\in X\backslash S$.
\end{remark}

\subsection{Comparison of Optimistic and Pessimistic Algorithms}
\emph{Tightness of Bounds}:
We can compare the performances of the optimistic and pessimistic algorithms by comparing the exponents in~\eqref{opt_curve_corollary} and~\eqref{pess_corollary}. Note that the algorithms will have the best approximation bound if the exponents evaluate to $-1$. The performance of each algorithm is dependent on the corresponding notions of curvature. For the optimistic algorithm, we wish that the $2$-marginal curvatures are close to zero for each $x_i$ and $S_{i-1}$. For the pessimistic algorithm we instead wish that the $2$-cardinality curvature is close to zero. One, downside that the pessimistic algorithm has is that each term of the sum has the $2$-cardinality curvature multiplied by $(i-1)$. This means that when the cardinality constraint $n$ is large, the $\min\{(i-1)\tau_2,1\}$ will saturate and resulting in the later terms of the sum to be evaluated to $0$, hindering the guaranteed performance of the pessimistic algorithm. Figure~\ref{fig:example_optimisitic_better} depicts an example where the optimistic algorithm will have tighter performance bounds, and Figure~\ref{fig:choices} conversely shows an example where the pessimistic algorithm will have tighter performance bounds provided by on our notions of curvature.

\emph{Assumptions Required}: It is important to note that the pessimistic algorithm requires that the function $f$ possesses the property of supermodularity of conditioning. This is a strong assumption on the functions and limits the number of applications the pessimistic algorithm can be applied. The optimistic algorithm on the other hand can be applied to arbitrary submodular functions. An advantage of the pessimistic algorithm, is that the performance bound is computable using only pairwise information. The performance bounds of the optimistic requires the ability to compute the objective function on sets of arbitrary size.

\emph{Empirical Results:}
As we show in our experimental results in Section~\ref{Simulation_results}, for the particular problem we explore, the pessimistic algorithm tends to out perform the optimistic algorithm in terms of approximation performance. Our experiments by no means show how the algorithms perform in every situation but highlights the potential achievable performance of the two pairwise algorithms.

\section{Time Complexity Comparison of Pairwise Algorithms to Classical Greedy}
\label{time_complexity}
A practical issue with the classical greedy strategy is that for large problems, the strategies are expensive to compute~\cite{LazyGreedy}. It is a common assumption in the literature that we have a \textit{value oracle} for the submodular objective function that is computable in constant time. This assumption leads to a time complexity of $\mathcal{O}(|X|\cdot n)$. In practical applications, the objective must be computed using a polynomial-time algorithm. This would result in a time complexity of $\mathcal{O}(|X|\cdot n\cdot T_{\text{eval}}(n))$ where $T_{\text{eval}}$ is cost of computing $f(S)$ given the size of $S$.

The pairwise algorithms can exploit evaluation of $f(A)$ only on sets $A\subseteq{X}$ and $|A|\leq 2$ to get significant time complexity improvements. The cost computing $f(x|y)$ for $x,y\in X$ is constant because it is not a function of the size of $X$ or $n$. 

Let $S = \{x_1,\dots,x_n\}\subseteq X$ be the set selected by a pairwise algorithm and $S_i = \{x_1,\dots, x_i\}$. The pairwise algorithms obtain their performance gains from the fact we can write $\bar{f}(x|S_{i})$ and $\underline{f}(x|S_{i})$ recursively. For the upper bound we have that for $x\in X$,
\begin{align}
    \bar{f}(x|S_{i}) & = \min_{x_j\in S_i}\{f(x|x_j)\}\nonumber\\
    & = \min\{\min_{x_j\in S_{i-1}}\{f(x|x_j)\},f(x|x_i)\}\nonumber\\
    & = \min\{\bar{f}(x|S_{i-1}),f(x|x_i)\}
    \label{ub_update}
\end{align}
Similarly, for each $x\in X$ for the lower bound have the following,
\begin{align}
    \underline{f}(x|S_{i})  &= f(x) - \sum_{x_j\in S_i}f(x)- f(x|x_j)\nonumber\\
    & = f(x) - \sum_{x_j\in S_{i-1}}f(x)-f(x|x_j) - (f(x)-f(x|x_i))\nonumber\\
    & = \underline{f}(x|S_{i-1}) - (f(x)-f(x|x_{i})).
    \label{lb_update}
\end{align}
Finding $\bar{f}(x|S_i)$ and $\underline{f}(x|S_i)$ can be computed in constant time as a function of $f(x|x_i)$ and $\bar{f}(x|S_{i-1})$ or $\underline{f}(x|S_{i-1})$ respectively. We can leverage this fact to compute each iteration of the pairwise greedy algorithms efficiently. For each iteration of the optimistic algorithm we compute the set  $E_i = \{\bar{f}(x|S_{i-1}): x \in X\backslash S_{i-1}\}$. We compute $x_i$ by finding the element with maximum value in $E$. We can compute $E_{i+1}$ via the recursive definition \eqref{ub_update} which can be done in linear time. Since we can compute both $E_{i+1}$ and the maximum of $E_{i}$ in linear time with respect to the size of $X$, the cost of each iteration of the algorithm is $\mathcal{O}(|X|)$. The exact same procedure can be done for the pessimistic algorithm. Therefore both the pairwise algorithms have time complexity of exactly $\mathcal{O}(|X|\cdot n)$.

The efficiency of the pairwise algorithms introduces a new trade off for practical applications of submodular maximization. The pairwise algorithms can avoid paying the cost $T_{\text{eval}}(n)$ while computing each marginal return. We can trade off approximation performance guarantees for Problem \ref{basicproblem} for execution speed improvements. This can be useful in scenarios where a user needs to repeatedly and quickly obtain approximate solutions to a submodular maximization problem, and are not that sensitive to the quality of the solution.  As we will show in our experimental results, the execution time improvements in using the pairwise algorithms can be significant while simultaneously still providing relatively strong approximation performance. If the function being maximized has favorable curvature conditions then the losses in the guaranteed performance from using the pairwise algorithms can be minimal.  

\section{Simulation Results}
\label{Simulation_results}
In this section, we benchmark the proposed algorithms in a simulated application of providing autonomous ride service in New York City utilizing electric vehicles. We focus on a coverage problem of selecting a set of charging locations for vehicles for which they can best respond to customer demand. From a historical data set provided by the NYC Taxi \& Limousine Commission \cite{NYCDATA}, we know that throughout the day the geographical distribution of customer demand is changing. 

\begin{figure}[ht]
    \centering
    \includegraphics[width = 0.3\textwidth]{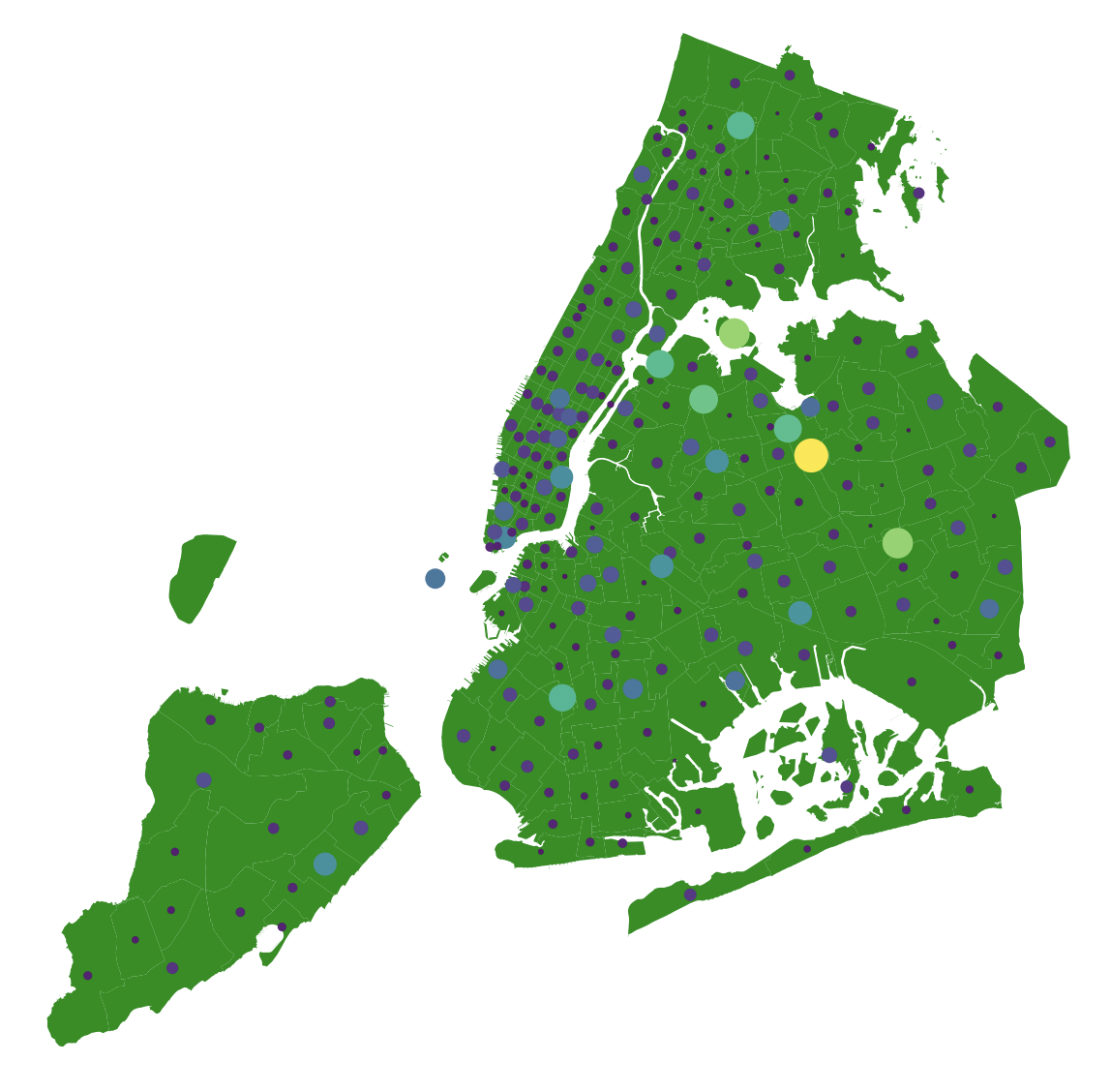}
    \caption{Geographical visualization of taxi customer data collected on January 1, 2020. There are 263 districts and size of the dots in each districts are proportional to the number of pick ups that occurred in the district. }
    \label{fig:nyc_data}
\end{figure}

New York City is split into 263 Taxi districts and we assume that the charging stations are located at the centroids of each of these districts. We wish to select a subset of charging locations that maximize the expected customer demand that can efficiently serviced from these locations.  We say a customer can be efficiently serviced if it can be picked up with a delay of at most $t$ minutes. Let $\mathcal{E}$ be the set of districts and let $X$ be the set of stations. Let $p_x^e$ be the probability that a vehicle deployed from station $x\in X$ can pick up a passenger in district $e\in \mathcal{E}$ in $t$-minutes. For simplicity we assume that the ride requests originate from the centroids of the districts. Finally, let $v_e$ be the demand in district $e$, which is modeled as the estimated number of pick up requests in district $e$ in a specified time interval, based on historical data.

Let $S\subseteq X$ be a set of stations. Then the objective we want to maximize, which we will call the hidden objective function $f_h$, is written as follows:
\begin{equation}
f_h(S)=\sum_{e \in \mathcal{E}}\left(\left(1-\prod_{x \in S}(1-p_{x}^{e})\right)v_{e}\right).
\end{equation}
The hidden objective function is more formally known as the probabilistic coverage function and was used for a related sensor coverage problem in \cite{distsubmaxpartition}. 

Suppose, we do not have access to the entire hidden objective function due to computation and/or modelling challenges, and thus the optimization is solved using only pairwise information. Given the pairwise information constraint, we can compute the expected demand that can be serviced by a single station and a pair of stations as,
$$f(x) = \sum_{e\in \mathcal{E}}p_x^e v_e$$
and 
$$f(x,y) = \sum_{e\in \mathcal{E}}(1-(1-p_x^e)(1-p_y^e))v_e.$$

We will model $p_{x}^e$ using a Gaussian Kernel function $$p_x^e = e^{-\frac{d(x,e)^2}{r_s^2}},$$
where $d:\real^2\times\real^2 \rightarrow \real_{+}$ is a distance metric, and $r_s$ is a tune-able parameter that dictates the range of distances where a vehicle could be quickly deployed to service a ride. For our experiments, we used the Euclidean distance metric for simplicity but the metric could be changed to better model the real system.

The objective function is a normalized monotone submodular function that exhibits supermodularity of conditioning \cite{distsubmaxpartition}, which allows us to apply all of our results. To decide which stations should be selected during different time intervals throughout the day, we estimate $v_e$ for each $e\in \mathcal{E}$ from the historical data and then attempt to solve $$S^* \in \argmax_{S\subseteq X, |S| \leq n}f_h(S).$$ 

We compare the optimistic and pessimistic algorithms ability to maximize $f_h$ while only given access to $f(x)$ and $f(x,y)$ for $x,y\in X$, to the full information greedy algorithm with full access to $f_h$. We compute the true objective $f_h(S)$ for the pairwise algorithms to compare against the full information greedy algorithm's performance.

\subsection{Approximation Performance Experimental Results}

To determine how the strategies of the optimistic and pessimistic strategies perform using historical data for ride services in New York City. The data set used included the pick-up times and locations of all the ``For Hire Vehicle" rides in the month of January 2020~\cite{NYCDATA}. We tested the performance of algorithms on varying distributions, we split the data up by pick-up time. We made twelve sections each corresponding to a unique two-hour window of the day. For each of the subsets, we estimated values $v_e$ for each district by taking the average number of rides in the district. We executed the three algorithms on each of the twelve sets to compare results, which are summarized in Figure~\ref{fig:classical_vs_lb_and_ub}. 
\begin{figure}[ht]
    \centering
    \includegraphics[width = 0.45\textwidth]{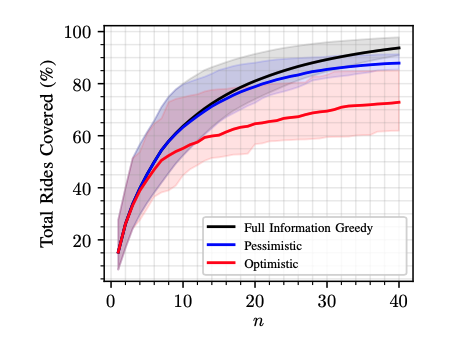}
    \caption{Experimental results comparing the performance of the three algorithms. The plot looks at percentage of total rides covered after selecting $n$ stations averaged over the twelve experiments. }
    \label{fig:classical_vs_lb_and_ub}
\end{figure}

\begin{figure}[ht]
    \centering
    \includegraphics[width = 0.45\textwidth]{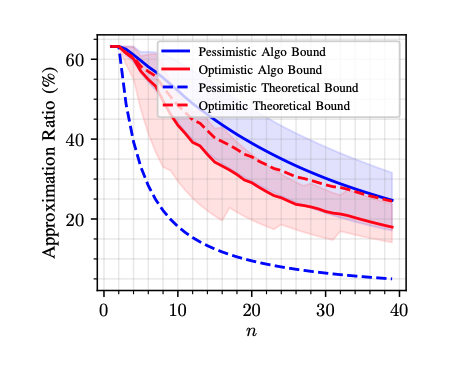}
    \caption{The estimated worst case lower bounds on performance of the optimistic and pessimistic algorithms. The bounds with solid lines were computed using Algorithm \ref{algo:posthoc}. The coloured filled sections show the maximum and minimum values of the bounds over all of the trials produced by Algorithm \ref{algo:posthoc}. The dashed lines are the average worst case bound over the trials computed using Corollaries \ref{ub_corollary_} and \ref{lb_corollary}.}
    \label{fig:estimated_bounds_lb_vs_ub}
\end{figure}
Figure~\ref{fig:classical_vs_lb_and_ub} shows the performance for different numbers of charging stations selected. We see that all three algorithms have similar performance for low values of $n$ but then begin to diverge after 15 stations are selected. The pessimistic strategy performs significantly better than the optimistic strategy. The pessimistic algorithm yielded a value no worse than 90\% of the full information greedy algorithm's value and the optimistic yielded a value no worse than 67\% across all trials. 
\begin{figure}[ht]
    \centering
    \includegraphics[width = 0.45\textwidth]{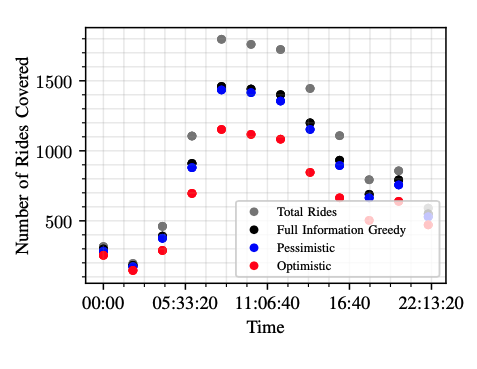}
    \caption{The number of rides covered by each algorithm during each two hour time interval of the day with n = 25 as well as the total possible rides covered at the time.}
    \label{fig:time_of_day}
\end{figure}

Figure~\ref{fig:estimated_bounds_lb_vs_ub} compares the worst-case lower bounds of the optimistic and pessimistic algorithms computed from the estimated approximation factors in \eqref{post_hoc_alpha} on the same trials as used in Figure \ref{fig:classical_vs_lb_and_ub} as well as the worst case performance bounds produced by Corollaries \ref{ub_corollary_} and \ref{lb_corollary}. As the number of stations selected increases, the lower bounds on performance degrade in all cases. For the pessimistic algorithm, this is due to the fact that the lower bounds on the marginals also exhibit diminishing returns and continually selecting elements maximizing the lower bound drives the denominator of \eqref{post_hoc_alpha} down. This causes high values of the estimated $\alpha_1,\dots,\alpha_n$ and decreasing approximation bounds. The optimistic algorithm does not actively minimize the estimated approximation factors, which are reflected in both the percentage of rides covered and approximation bounds. As observed in Figure~\ref{fig:classical_vs_lb_and_ub}, we see that both pairwise algorithms are still performing similarly to the full information algorithm even though the lower bounds in Figure~\ref{fig:estimated_bounds_lb_vs_ub} suggests otherwise. We also see that the theoretical performance bound for the optimistic algorithm is relatively close to the bound produced by Algorithm \ref{algo:posthoc} for the pessimistic algorithm. The pessimistic algorithm's theoretical bound is much lower than the rest of the bounds, due to the fact that the average $\tau_2$ for each trial was near $1$ with a value of $0.89$. This highlights the benefits of utilizing Algorithm \ref{algo:posthoc} when computing performance bounds for the pessimistic algorithm.

This experiment reveals that the bound produced by Algorithm \ref{algo:posthoc} becomes less accurate as $n$ increases. This is due to the fact that the approximation factors measure the multiplicative difference from the true greedy choices. The true marginals for elements selected near the end of execution tend to be smaller. Thus, the difference in the overall objective values could be small, but the multiplicative difference could be large which is reflected in the lower bound on performance.

Figure~\ref{fig:time_of_day} looks at the performance of the algorithms over different subsets of the historical data. We plotted the objective value for 25 stations using each of the algorithms for each of the twelve subsets of data. We also plotted the max value of the objective function each of the algorithms could possibly achieve in the time frame. From Figure \ref{fig:time_of_day}, we see that for each time frame, the pessimistic algorithm is essentially as effective as the full information greedy algorithm but the optimistic greedy algorithm is less effective. This shows that our function $f$ is not an example of a function where the performance is near the universal lower bound described in Theorem \ref{approx_theorem}:  Both algorithms are performing near the full information greedy strategy which is at least $63$\% of optimal.  This is far from approximately 8\% of optimal as dictated by Theorem \ref{approx_theorem}.

\subsection{Time Complexity Experimental Results} 
Using the same data used for the performance experiments, we measured the time efficiency of the three algorithms as well. We measured the execution time of the algorithms on each of the 12 data subsets and plotted the average execution time in terms of $n$. Each of the experiments was computed using Python 3.7 on a 2017 Macbook Pro with a 3.1 GHz Dual-Core Intel i5 with 8 GB 2133 MHz LPDDR3 RAM.

\begin{figure}[ht]
    \centering
    \includegraphics[width = 0.45\textwidth]{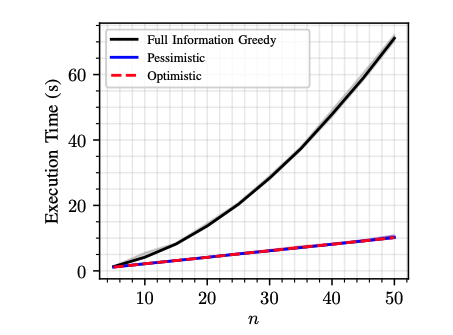}
    \caption{Execution time of algorithms as number of stations selected increases. For each $n$, the algorithms execution times were recorded on each of the 12 subsets of data and then averaged over 5 trials. The coloured filled sections are the maximum and minimum execution for each of the algorithms for each $n$.}
    \label{fig:execution_time}
\end{figure}

\begin{figure}[ht]
    \centering
    \includegraphics[width = 0.45\textwidth]{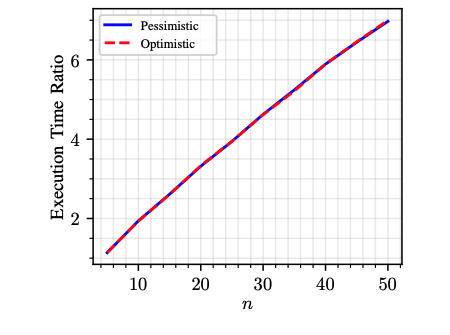}
    \caption{Plot of the ratios of execution time of the greedy strategy to the execution time of pairwise algorithms as number of stations selected increases. For each $n$, the average execution time of the the greedy algorithm was divided by the execution time of the algorithms.}
    \label{fig:execution_ratio}
\end{figure}
Figure~\ref{fig:execution_time} summarizes the results from the execution time experiment. For each trial, the size of $|X|$ was the same. The relationship between the execution times of the pairwise algorithms and the value of $n$ is linear. This relationship is as expected given our time complexity analysis in Section \ref{time_complexity}. Using the pairwise greedy strategies and the implementation details from Section~\ref{time_complexity}, reduces the time complexity from quadratic to linear time in terms of $n$. The pessimistic and optimistic algorithms share similar execution times which resulted in the two lines overlapping.

In Figure \ref{fig:execution_ratio}, we see that for both pairwise algorithms the ratio of the execution time of the full information greedy algorithm and the pairwise algorithms is almost linear. This verifies that the pairwise algorithms result in a reduction in time complexity by almost a factor of $n$ for this particular objective function.

\section{Conclusions and Future Work}

We have studied the problem of submodular maximization when an algorithm has limited access to the objective function. In general, without any additional assumptions about the function an algorithm cannot guarantee better performance than $k/n$ given $k$-wise information about the objective. We proposed two strategies that utilizes only pairwise information, which have performance guarantees dependent on new notions of curvature. We also provided a method to measure the performance of algorithms with limited access to the objective function in hindsight under the assumption that the function possesses supermodularity of conditioning. The two proposed strategies can be computed efficiently and provide new trade offs between approximation performance and time complexity. We also provided experimental results highlighting the performance of both algorithms.

In future work, we plan on extending these results to scenarios in distributed systems where a team of agents can only observe a subset of the other agents' decisions as well as partially observe the objective values. We also are exploring how $k$-wise information can be leveraged in similar style as pairwise lower bound to produce new limited information greedy strategies.

% Generated by IEEEtran.bst, version: 1.14 (2015/08/26)

\newpage
\section{Appendix}
The following is a proof of Theorem \ref{approx_theorem}.
\begin{proof}
Let $S^*$ be the solution to Problem \ref{basicproblem}. Recall Proposition 2.1 from \cite{Nemhauser} that $f$ is a monotone submodular set function on $X$ if and only if $f(T) \leq f(S) + \sum_{x_j\in T\backslash S}f(x_j|S)$ for all $S,T\subseteq X$. If $S$ is empty then we have 
\begin{equation}
     f(S^*) \leq \sum_{x_i^*\in S^*}f(x_i^*) \leq nf(x_1^*)\leq n\alpha_1f(x_1)\label{g_lb_theorem:1}.
\end{equation}
Then if we apply the lemma again with $S_j$, we have 
\begin{align}
    f(S^*) \leq f(S_j) +\sum_{x_j^*\in S^*\backslash S_j}f(x_j^*|S_j).\label{g_lb_theorem:2}
\end{align}
We also know that
\begin{align}
    \alpha_{j+1}f(x_{j+1}|S_j)\geq \max_{x\in X\backslash S_j}f(x|S_j)\geq f(x_j^*|S_j).\label{g_lb_theorem:3}
\end{align}
We now substitute equation \eqref{g_lb_theorem:3} into equation \eqref{g_lb_theorem:2} and write $f(S_j)$ as sum of it's marginals to get 
\begin{align}
     f(S^*) &\leq f(S_j) +\sum_{x_j^*\in S^*\backslash S_j}\alpha_{j+1}f(x_{j+1}|S_j)\nonumber\\
     &\leq \sum_{i=1}^{j}f(x_i|S_{i-1}) +n\alpha_{j+1}f(x_{j+1}|S_j) \label{g_lb_theorem:4}
\end{align}
Equation \eqref{g_lb_theorem:4} holds since $|S^*\backslash S_j|\leq n$. We will now rearrange equation \eqref{g_lb_theorem:4} to get the following
\begin{equation}
    f(x_{j+1}|S_j)\geq \frac{1}{\alpha_{j+1}n}f(S^*) - \frac{1}{\alpha_{j+1}n}\sum_{i=1}^{j}f(x_i|S_{i-1})\label{g_lb_theorem:5}.
\end{equation}
Now we will add $\sum_{i = 1}^jf(x_i|S_{i-1})$ to both sides of equation \eqref{g_lb_theorem:5} and simplify
\begin{align}
    \sum_{i=1}^{j+1}f(x_{i}|S_{i-1}) \label{g_lb_theorem:6}\geq \frac{1}{\alpha_{j+1}k}&f(S^*) \nonumber\\
    & + \frac{\alpha_{j+1}k-1}{\alpha_{j+1}k}\sum_{i=1}^{j}f(x_i|S_{i-1}) 
\end{align}
We will now prove by induction on $j$ that 

\begin{align}
    \sum_{i=1}^{j}f(x_{i}|S_{i-1})\geq \frac{\prod_{i=1}^{j}(\alpha_in)-\prod_{i=1}^{j}(\alpha_in-1)}{\prod_{i=1}^{j}(\alpha_in)}f(S^*)\nonumber
\end{align}

For base case $j = 1$ we will apply equation \eqref{g_lb_theorem:1} to get
\begin{align}
        f(x_1)& \geq \frac{1}{\alpha_1n}f(S^*) \nonumber
\end{align}
proving the base case. Now assuming the claim holds for $j-1$. We will apply the inductive hypothesis to equation \eqref{g_lb_theorem:6}, 
\begin{align}
    \sum_{i=1}^{j}&f(x_{i}|S_{i-1})\geq \frac{1}{\alpha_{j}n}f(S^*) \nonumber\\
    &+\frac{\alpha_{j}n-1}{\alpha_{j}n}\cdot\frac{\prod_{i=1}^{j-1}(\alpha_in)-\prod_{i=1}^{j-1}(\alpha_in-1)}{\prod_{i=1}^{j-1}(\alpha_in)}f(S^*)
\end{align}
Then after rearranging, we arrive at 
\begin{align}
    f(S_j)&\geq \frac{\prod_{i=1}^{j}(\alpha_in)-\prod_{i=1}^{j}(\alpha_in-1)}{\prod_{i=1}^{j}(\alpha_in)}f(S^*)\nonumber
\end{align}
proving the inductive hypothesis. If we take $j = n$, we arrive at
\begin{align}
    f(S_n)&\geq \frac{\prod_{i=1}^{n}(\alpha_in)-\prod_{i=1}^{n}(\alpha_in-1)}{\prod_{i=1}^{n}(\alpha_in)}f(S^*). \nonumber
\end{align}
We will now lower bound right coefficients on $f(S^*)$ to simplify the bound. We can now cancel out the denominator of the coefficient to get.
\begin{align}
   \frac{\prod_{i=1}^{n}(\alpha_in)-\prod_{i=1}^{n}(\alpha_in-1)}{\prod_{i=1}^{n}(\alpha_in)} & = 1-\prod_{i=1}^n\frac{\alpha_i n-1}{\alpha_i n}\\
   & = 1-\prod_{i=1}^n\left(1 -\frac{1}{\alpha_i n}\right)
\end{align}
We can now upper bound each term in the product using $1+x \leq e^x$ with $x = \frac{1}{\alpha_i n}$ to get a lower bound.
\begin{align}
    1-\prod_{i=1}^n\left(1 -\frac{1}{\alpha_i n}\right) &\geq 1-\prod_{i=1}^ne^{-\frac{1}{\alpha_i n}}\\
    &= 1-e^{-\frac{1}{n}\sum_{i=1}^n\frac{1}{\alpha_i}}
\end{align}
Using this lower bound yields our result.
\end{proof}

Proof of Theorem \ref{k_wise_ub_approx_theorem}
\begin{proof}
Let $x_i^g$ be the true greedy choice at iteration $i$ given $S_{i-1}$. For $i\leq k$ we have that $x_i = x_i^g$ by the definition of $\bar{f}_k(x_i|S_{i-1})$. Therefore we have, $\alphaoptk{1} = \dots = \alphaoptk{k} = 1$. The minimum possible approximation factor we have can be written as 
\begin{equation}
    \alpha_i^{\min} = \frac{f(x_i^g|S_{i-1})}{f(x_i|S_{i-1})}
\end{equation}
Any approximation factor $\alpha_i$ such that $\alpha_i\geq \alpha_i^{\min}$ will satisfy equation \eqref{approx_factor_def}.  We will now upper bound $\alpha_i^{\min}$ as follows.

\begin{align}
    \alpha_i^{\min} &= \frac{f(x_i^g|S_{i-1})}{f(x_i|S_{i-1})}\nonumber\\
    &\leq \frac{\bar{f}_k(x_i^g|S_{i-1})}{f(x_i|S_{i-1})}\label{k_by_ub_def}\\
    &\leq \frac{\bar{f}_k(x_i|S_{i-1})}{f(x_i|S_{i-1})} \label{k_by_greedy_choice}
\end{align}
Where equation \eqref{k_by_ub_def} holds by the definition of the upper bound. Equation \eqref{k_by_greedy_choice} holds by the greedy choice property of the $k$-wise optimistic algorithm. Therefore the right hand side of equation \eqref{k_by_greedy_choice} is a valid approximation factor. We set $\alphaoptk{i}$ to the right hand side of equation \eqref{k_by_greedy_choice} we conclude our proof.
\end{proof}

Proof of Theorem \ref{lb_approx_theorem}.
\begin{proof}
Let $x_i^g$ be the true greedy choice at iteration $i$ given $S_{i-1}$. For $i = \{1,2\}$, we have that 
\[
\underline{f}(x|S_{i-1}) = f(x|S_{i-1}).
\]
Hence, $x_i = x_i^g$ and therefore, we can let $\alphapes{1} = \alphapes{2} = 1$. The minimum approximation factor that we can achieve can be written as follows.
\begin{align}
    \alpha_i^{\min} &=\frac{f(x_i^g|S_{i-1})}{f(x_i|S_{i-1})}\nonumber\\
    &\leq \frac{f(x_i^g|S_{i-1})}{\underline{f}(x_i|S_{i-1})}\label{by_lb_def}\\
    &\leq \frac{f(x_i^g|S_{i-1})}{\underline{f}(x_i^g|S_{i-1})}\label{greedy_choice}\\
    &\leq \frac{f(x_i^g|S_{i-1})}{f(x_i^g)(1-\min\{(i-1)\tau_2,1\})}\label{curv_lb}\\
    &\leq \frac{1}{1-\min\{(i-1)\tau_2,1\}}\label{pess_factor}
\end{align}
Where \eqref{by_lb_def} holds by the definition of the lower bound, \eqref{greedy_choice} holds by the greedy choice property of the pessimistic strategy and finally \eqref{curv_lb} holds by  \eqref{lb_curve_lowerbound}. Therefore if we set $\alphapes{i}$ to the right hand side of equation \eqref{pess_factor}, then $\alphapes{i}$ is a valid approximation factor for Theorem \ref{approx_theorem}.
\end{proof}

\end{document}